\documentclass[11pt,letterpaper]{article}
\usepackage[margin=1in]{geometry}
\usepackage{ifpdf}
\ifpdf
    \usepackage[pdftex]{graphicx}
    \graphicspath{{figs/}{matlab/}}   
    \usepackage[update]{epstopdf}
\else
	\usepackage{graphicx}
	\graphicspath{{figs/}{matlab/}}   
\fi
\usepackage{pdflscape}
\usepackage{marvosym}
\usepackage{enumitem,color}
\usepackage{amsthm}
\newtheorem{theorem}{Theorem}

\newtheorem{lemma}{Lemma}

\usepackage{hyperref}
\usepackage{array}
\usepackage{amsmath,amsthm}
\usepackage{amsfonts}
\usepackage{amssymb}
\usepackage{cite,setspace}
\usepackage{multirow}
\usepackage{authblk}
\usepackage{appendix}

\newtheorem{definition}{\textbf{Definition}}

\newcommand{\Expt}{\mbox{${\mathbb E}$} }

\begin{document}

\title{Capacity-Achieving Private Information Retrieval Codes with Optimal Message Size and Upload Cost}

\author{Chao Tian \qquad Hua Sun \qquad Jun Chen}

\maketitle

\begin{abstract}
We propose a new capacity-achieving code for the private information retrieval (PIR) problem, and show that it has the minimum message size (being one less than the number of servers) and the minimum upload cost (being roughly linear in the number of messages) among a general class of capacity-achieving codes, and in particular, among all capacity-achieving linear codes. Different from existing code constructions, the proposed code is asymmetric, and this asymmetry appears to be the key factor leading to the optimal message size and the optimal upload cost. The converse results on the message size and the upload cost are obtained by an analysis of the information theoretic proof of the PIR capacity, from which a set of critical properties of any capacity-achieving code in the code class of interest is extracted. The symmetry structure of the PIR problem is then analyzed, which allows us to construct symmetric codes from asymmetric ones, yielding a meaningful bridge between the proposed code and existing ones in the literature. 
\end{abstract}

\section{Introduction}

The private information retrieval (PIR) problem addresses the following scenario. A total of $K$ messages, each of $L$ bits (or $L$ symbols in some finite alphabet), are replicated at $N$ servers. A user wishes to retrieve one of the messages without revealing the identity of the desired message to any individual server. To retrieve this message, the user generates one query for each server and each server will return an answer to the user, which depends on the stored messages and the received query. To ensure that each server learns nothing about which message is being retrieved in the information theoretic sense, each query must be marginally independent of the desired message index. The PIR problem admits a trivial solution, where the user simply requests all the messages. However, downloading everything obviously incurs too much communication cost, and PIR systems should be designed to communicate as efficiently as possible between the user and the servers. 

In PIR systems, the most important measure of communication efficiency is the retrieval rate, defined as the number of message information bits that can be retrieved per bit of downloaded data from the servers. The maximum value of retrieval rate of a PIR system is referred to as its capacity, and the problem of characterizing the capacity is of fundamental importance in this setting. This problem was recently settled in \cite{sun2017PIRcapacity} where the capacity was found to be 
\begin{align}
C =\left(1+\frac{1}{N}+\frac{1}{N^2}+\ldots+\frac{1}{N^{K-1}}\right)^{-1}. \label{eqn:capacityfomular}
\end{align} 
Other notable efforts and generalizations on the PIR problem in the coding and information theory literature can be found in, {\em e.g., }\cite{banawan2018capacity,Chan_Ho_Yamamoto,Blackburn_Etzion_Paterson,Tajeddine_Rouayheb,FREIJ_HOLLANTI,Zhang_Wang_Wei_Ge,Shah_Rashmi_Kannan,
Fazeli_Vardy_Yaakobi,Sun_Jafar_MPIR,Sun_Jafar_MDSTPIR,Wang_Skoglund,Kumar_Rosnes_Amat,Tandon_CachePIR,Banawan_Ulukus_Byzantine,
Banawan_Ulukus_Multimessage,Banawan_Ulukus_Asymmetric,Wang_Skoglund_SPIREve,Lin_Kumar_Rosnes_Amat,Tajeddine_Gnilke_Karpuk_Etal,sun2017optimal,Tian_Sun_Chen_Storage}. 

In the previous works where the PIR capacity is concerned, such as \cite{sun2017PIRcapacity, banawan2018capacity, Wang_Skoglund, Tandon_CachePIR,Banawan_Ulukus_Multimessage,Banawan_Ulukus_Asymmetric,Wang_Skoglund_SPIREve}, it is usually assumed that the message length $L$ is sufficiently large ($L$ is allowed to go to infinity). As a consequence, the corresponding code constructions in the literature are usually built by recursively layering message symbols and parity symbols using symmetry relations, resulting in codes that can only be applied on very long messages. The number of symbols in each message that a code can be applied on is sometimes referred to as the sub-packetization factor of the code. A smaller message length (sub-packetization factor) means that the code is more versatile, has less constraints, and may lead to more efficient implementation in practice. Another design factor that is of practical importance is the possible number of queries that each server needs to accommodate, {\em i.e., } the cardinalities of the query sets. Small cardinalities of the query sets imply that, firstly, the amount of information that needs to be sent to the servers (often referred to as the upload cost) is small during the query operation, and secondly, the servers only need to compute a small set of functions, both of which lead to simpler and more efficient system implementation. 

In this work, we consider the construction of capacity-achieving PIR codes, and the contribution of this work is three-fold. 
\begin{enumerate}[itemsep=0pt]
\item Firstly, we propose a novel capacity-achieving PIR code construction, which has a small message size of $(N-1)$ bits and a low upload cost of $N(K-1)\log_2N$ bits. The coding alphabet of the proposed code can in fact be chosen to be any finite group or finite field, and particularly, to be the binary field which results in a binary linear code. Different from existing code constructions, the code proposed in this work is asymmetric, and this asymmetry appears to be the key to the significant reductions in the message size and the upload cost compared to other capacity-achieving codes. 
\item Secondly, through a novel and delicate analysis of the converse proof of the PIR capacity, we identify a set of critical properties for a class of capacity-achieving codes on abelian groups, which we refer to as decomposable codes. Based on these properties we further derive novel converses for the message size and the upload cost. These converse bounds match the corresponding values in the proposed code, thus establishing the optimality of the proposed code construction in terms of the message size and the upload cost within the corresponding code classes, in particular, among capacity-achieving linear codes. 
\item Last but not least, the relation between symmetric PIR codes and asymmetric PIR codes is analyzed in details. The symmetry in this problem setting in fact includes three different components, namely symmetry in server indices, symmetry in file indices, and symmetry in the query answers. The analysis reveals certain fundamental structures in the problem setting that were largely overlooked in the existing literature. Using these symmetry relations, we show that the proposed code (in fact any asymmetric code) can be used to build more symmetric PIR codes, which offers a bridge between the proposed code and the existing code constructions. 
\end{enumerate}

It should be noted that efforts on the PIR problem in the theoretical computer science community focus on an alternative formulation where the message length $L$ is assumed to be small and fixed, usually a single bit, but the number of messages and the number of servers are allowed to grow asymptotically \cite{PIRfirstjournal}. In this setting, the overall communication cost can be viewed as consisting of upload cost and download cost, the latter of which is inversely proportional to the retrieval rate, and they can be traded off between each other. In fact, there exists a complex relation among the three quantities of the message size, the upload cost, and the download cost. For example, for the solution of retrieving everything, the upload cost is 0 as nothing needs to be sent to the servers ({\em i.e.}, there is no randomness in the queries) and the download cost is $KL$ as all messages are retrieved, and the message size can be $L=1$ bit each. Characterizing even the sum cost of upload and download for the case $L=1$ in the original theoretical computer science formulation appears to be intractable, and instead order-wise bounds have been investigated; there have been considerable efforts and many significant results after the ground-breaking work of \cite{PIRfirstjournal}; see {\em e.g.,} \cite{Ambainis,Beimel_Ishai_Kushilevitz_Raymond,Efremenko, Yekhanin_LDC, 2PIR,Ishai_Kushilevitz}. Against this general backdrop, our result can be viewed as the first to precisely determine the relation among the message size, the upload cost, and the download cost, for the extreme point when the download cost is minimized.   

The rest of the paper is organized as follows. Section \ref{sec:model} provides the problem definition and the necessary notation. Section \ref{sec:newPIR} gives the proposed PIR code construction. The converse results on the minimum message size and the minimum upload cost are given in Section \ref{sec:optimal}, and the symmetry relations are discussed in Section \ref{sec:symmetry}. Finally, Section \ref{sec:conclusion} concludes the paper.

\section{Model and Preliminaries}
\label{sec:model}

In this section, we provide a formal problem definition, as well as the necessary notation for subsequent discussions. A slightly different indexing method is chosen in this work: instead of the more conventional indexing of starting at $1$, the indexing here starts at $0$. This does not make any essential difference in the problem and the solution, however it will lead to notional simplicity when we present the new code construction. 

\subsection{System Model}
The private information retrieval model can be formally described as follows. There are a total of $N$ servers, each storing a copy of $K$ messages, denoted as $W_0,W_1,\ldots,W_{K-1}$, respectively. A user wishes to retrieve a message $W_k$, $k\in \{0,1,\ldots,K-1\}$, however at the same time wishes to keep the identity of the message being retrieved as a secret to any one of the servers. For this purpose, the user, using a random input  $\mathsf{F}$ as the key, chooses a set of queries, $Q_{0:N-1}=(Q_0,Q_1,\ldots,Q_{N-1})$, one per server, and sends the queries to the servers. Server-$n$ responds with an answer $A_n$, which depends on the messages stored at the server and the received query. Using all the answers $A_{0:N-1}=(A_0,A_1,\ldots,A_{N-1})$ from all the servers, together with the values of $\mathsf{F}$ and $k$, the user then reconstructs $W_k$. The privacy requirement stipulates that at each server, the probability distributions on the allowed queries are identical for all the messages, thus the server cannot learn any information regarding which message is being requested.  

We now give a more mathematically precise description of the problem. Denote the set of possible queries for server-$n$ as $\mathcal{Q}_n$, and denote its cardinality as $|\mathcal{Q}_n|$. The cardinality of a set $\mathcal{A}$ will be similarly denoted as $|\mathcal{A}|$ in the rest of the paper. Assume that the random key $\mathsf{F}$ is uniformly distributed on a certain finite set $\mathcal{F}$. Moreover, a message $W_k$ consists of $L$ symbols, each symbol belonging to a finite alphabet $\mathcal{X}$; in particular, for messages in computer systems, we usually use $\mathcal{X}=\{0,1\}$. The messages are mutually independent, each of which is uniformly distributed on $\mathcal{X}^L$.  We further allow the query answers to be represented as a variable length vector, whose elements are in the finite alphabet $\mathcal{Y}$, though our code construction will eventually only use $\mathcal{Y}=\mathcal{X}$.

\begin{definition}
\label{def:problem}
An $N$-server private information retrieval (PIR) code for $K$ messages, each of $L$-symbols in the alphabet $\mathcal{X}$, consists of
\begin{enumerate}[itemsep=0pt]
\item $N$ query functions: 
\begin{align}
\phi_n: \{0,1,\ldots,K-1\}\times \mathcal{F}\rightarrow \mathcal{Q}_n,\quad n\in \{0,,1,...,N-1\},
\end{align}
i.e., the user chooses the query $Q^{[k]}_n=\phi_n(k,\mathsf{F})$ for server-$n$, using the index of the desired message and the random key $\mathsf{F}$;
\item $N$ answer length functions: 
\begin{align}\ell_n: \mathcal{Q}_n\rightarrow \{0,1,2,\ldots\},\quad n\in \{0,1,...,N-1\}, 
\end{align} 
i.e., the length of the answer at each server, a non-negative integer, is a deterministic function of the query, but not the particular realization of the messages; 
\item $N$ answer functions: 
 \begin{align}
\varphi_{n}: \mathcal{Q}_n\times \mathcal{X}^{KL}\rightarrow \mathcal{Y}^{\ell_n},\quad n\in\{0,1,...,N-1\},\label{eqn:answers}
\end{align}
where $\ell_n=\ell_n(q_n)$ with $q_n\in \mathcal{Q}_n$ being the (random) query for server-$n$, $\mathcal{Y}$ is the coded symbol alphabet, and in the sequel we shall write the query answer as $A^{[k]}_n\triangleq\varphi_{n}(Q^{[k]}_n,W_{0:K-1})$ when the message index $k$ is relevant;
\item A reconstruction function using the answers from the servers together with the desired message index and the random key: 
\begin{align}
\psi: \prod_{n=0}^{N-1} \mathcal{Y}^{\ell_n}\times \{0,1,...,K-1\}\times \mathcal{F}\rightarrow \mathcal{X}^{L},
\end{align} 
i.e., $\hat{W}_k=\psi(A^{[k]}_{0:N-1},k,\mathsf{F})$ is the retrieved message. 
\end{enumerate}
These functions should satisfy the following two requirements:
\begin{enumerate}

\item \textbf{Correctness: } For any $k\in \{0,1,...,K-1\}$, $\hat{W}_k=W_k$. 

\item \textbf{Privacy: } For every $k,k'\in\{0,1,...,K-1\}$, $n\in \{0,1,...,N-1\}$, and $q\in \mathcal{Q}_n$, 
\begin{align}
\mathbf{Pr}(Q^{[k]}_n=q)=\mathbf{Pr}(Q^{[k']}_n=q).
\end{align}
\end{enumerate}
\end{definition}

The correctness condition here requires that the reconstructed message as a random variable is the same as the requested message, and it thus inherently requires that for any realization of $\mathsf{F}$, the equality must hold.  It is in fact without loss of generality to restrict $\mathcal{F}$ and $\mathcal{Q}_n$'s to be certain finite sets of integers, however, we allow them to be more general sets, which will facilitate describing more concisely the proposed PIR code construction. It is also worth noting that the alphabet $\mathcal{Y}$ in the problem definition may be an abstract finite set, with no further structure assigned to it. However, for any such a finite set, we can establish a bijective mapping between $\mathcal{Y}$ and the set $\{0,1,\ldots,|\mathcal{Y}|-1\}$. By further enforcing an operation between any two elements in the latter set (for example, modulo $|\mathcal{Y}|$ addition), the set $\mathcal{Y}$ can also be assigned an operation through homomorphism. In other words, any abstract set $\mathcal{Y}$ can always be viewed as a finite group, however requiring $\mathcal{Y}$ to be a finite group in the problem definition is  unnecessary. 

\subsection{Two General Code Classes}

We next define precisely the code classes in which we establish the optimality of our proposed code construction. These definitions are technical, and the readers may wish to skip them at the initial read and simply consider the more restricted code class of vector linear codes on a finite field, without materially jeopardizing understanding the code construction in Section \ref{sec:newPIR}. These two definitions only become important in Section \ref{sec:optimal} when the optimality of the proposed code construction is established within these more general classes. 

\begin{definition}
\label{def:linearcodes}
A PIR code is called decomposable, if $\mathcal{Y}$ is a finite abelian group, and for each fixed $n\in \{0,1,\ldots,N-1\}$ and $q\in \mathcal{Q}_n$, the answer function $\varphi_{n}(q,W_{0:K-1})$ can be written in the form
\begin{align}
\varphi_{n}(q,W_{0:K-1})=\left(\varphi^{(q)}_{n,0}(W_{0:K-1}),\varphi^{(q)}_{n,1}(W_{0:K-1}),\ldots,\varphi^{(q)}_{n,\ell_n-1}(W_{0:K-1})\right),
\end{align}
where
\begin{align}
\varphi^{(q)}_{n,i}(W_{0:K-1})=\varphi^{^{(q)}}_{n,i,0}(W_0)\oplus\varphi^{(q)}_{n,i,1}(W_1) \oplus \ldots \oplus \varphi^{^{(q)}}_{n,i,K-1}(W_{K-1}),\quad i\in \{0,1,\ldots,\ell_n-1\}\label{eqn:decompose}
\end{align}
where $\oplus$ represents addition in the finite group $\mathcal{Y}$, and each $\varphi^{(q)}_{n,i,k}$ is a mapping $\mathcal{X}^L\rightarrow\mathcal{Y}$. 
\end{definition}

The terminology ``decomposable'' comes from (\ref{eqn:decompose}) which restricts each coded symbol to be a summation (in the abelian group) of the component functions on the individual messages. 
Let us consider an example where the two messages, each of a single symbol, are in certain ring $(\mathcal{Y},\oplus,\otimes)$, and the resulting answer for $n=q=0$ is 
\begin{align*}
\varphi_{0}(0,(W_0,W_1))&=(\varphi^{(0)}_{0,0}(W_0,W_1),\varphi^{(0)}_{0,1}(W_0,W_1),\varphi^{(0)}_{0,2}(W_0,W_1))\\
&=\big{(}(W_0\otimes W_0)\oplus \alpha,(W_0\otimes W_0)\oplus (W_1\otimes W_1)\oplus \alpha, (W_1\otimes W_1)\oplus \alpha\big{)},
\end{align*}
where $\alpha$ is an element of the ring. This code belongs to the code class of decomposable codes, but it is clearly not linear. In the component function $\varphi^{(0)}_{0,1}(W_0,W_1)$, we have
\begin{align*}
\varphi^{(0)}_{0,1}(W_0,W_1)=\varphi^{^{(0)}}_{0,1,0}(W_0)\oplus\varphi^{(0)}_{0,1,1}(W_1)=(W_0\otimes W_0)\oplus((W_1\otimes W_1)\oplus \alpha).
\end{align*}

\begin{definition}
\label{def:uniform}
If a decomposable PIR code has the property that any component function $\varphi^{(q)}_{n,i,k}$ in (\ref{eqn:decompose}) either satisfies the condition
\begin{align}
\left|\left\{w\in \mathcal{X}^L:\varphi^{(q)}_{n,i,k}(w)=g\right\}\right|=\left|\left\{w\in \mathcal{X}^L:\varphi^{(q)}_{n,i,k}(w)=g'\right\}\right|,\quad \forall g,g'\in \mathcal{Y},
\end{align}
or it maps everything to the same value,  i.e., 
\begin{align}
\varphi^{(q)}_{n,i,k}(w)=\varphi^{(q)}_{n,i,k}(w'),\quad \forall w,w'\in\mathcal{X}^L, 
\end{align}
then the PIR code is called uniformly decomposable.
\end{definition}

A uniformly decomposable PIR code has the property that the decomposed message mappings $\varphi^{(q)}_{n,i,k}$ will preserve a uniform probability distribution on the coded symbol alphabet, unless the induced random variable is in fact deterministic. 
The notion of decomposable codes considerably generalizes the notion of linear codes. In particular, linear codes on finite fields are uniformly decomposable, and linear codes defined on modules over a ring \cite{connelly2018linear1,connelly2018linear2} are decomposable (and some are uniformly decomposable); it also naturally includes codes defined on cosets of a binary lattice and some nonlinear codes.  
In Section \ref{sec:optimal}, we establish general outer bounds for decomposable codes, which imply that the proposed code is optimal in the corresponding code classes, and particularly, it is optimal among all linear codes.

Decomposable codes can be simply represented as 
\begin{align}
\varphi_{n}(q,W_{0:K-1})=W_{0:K-1}\cdot G^{(q)}_n, \label{eqn:matrix}
\end{align}
where $W_{0:K-1}$ is viewed as a length-$K$ vector whose components are in the alphabet $\mathcal{X}^L$, and $G^{(q)}_n$ is a matrix of dimension $K\times\ell_n$ whose elements $G^{(q)}_{n,i,k}$ are functions $\mathcal{X}^L\rightarrow\mathcal{Y}$ with the ``$\cdot$'' operation between $W_k$ and the matrix element $G^{(q)}_{n,i,k}$ defined as
\begin{align}
W_{k}\cdot G^{(q)}_{n,i,k}\triangleq\varphi^{(q)}_{n,i,k}(W_k).
\end{align}

Consider another example  (a uniformly decomposable code) where $K=3$, $L=1$, and the answer for query $q=0$ for server-0 is 
\begin{align}
\varphi_0(q=0,W_{[0:2]})&=\varphi^{(0)}_{0,0}(W_0,W_1,W_2)=W_{0}\ominus W_{2}
\end{align}
where $\mathcal{X}=\mathcal{Y}$ is the finite group $\{0,1,2,3\}$ with $\oplus$ and $\ominus$ being the modulo-$4$ addition and subtraction. Then we have 
\begin{align}
\varphi_0(q=0,W_{[0:2]})&=\varphi^{(0)}_{0,0,0}(W_0)\oplus \varphi^{(0)}_{0,0,1}(W_1)\oplus\varphi^{(0)}_{0,0,2}(W_2)=(W_{0})\oplus (0)\oplus(\ominus W_{2})\nonumber\\
&=[W_0,W_1,W_2]\cdot G^{(0)}_0=[W_0,W_1,W_2] \begin{bmatrix}1\\0\\-1\end{bmatrix},
\end{align}
In the matrix representation, $G^{(0)}_0=[1,0,-1]^{t}$, where $1$ stands for the identity function, $0$ for the all zero function, and $-1$ for the negation function.

For linear codes defined on a finite field $\mathcal{X}$, the function $\varphi^{(q)}_{n,i,k}(W_k)$ is the inner product between the length-$L$ message vector $W_k$, and a fixed length-$L$ coding coefficient vector in the same finite field. In this case, $W_{0:K-1}$ can be alternatively written as a length-$KL$ vector in the alphabet $\mathcal{X}$, and the matrix $G^{(q)}_{n}$ can be further expanded as a $KL\times\ell_n$ matrix whose elements are also in $\mathcal{X}$, and the finite field addition and multiplication will be used in the matrix multiplication. Such a matrix $G^{(q)}_{n}$ is in fact simply the familiar generator matrix of (vector) linear codes \cite{LinCostello:book}.

It should be noted that most converse results on linear codes in the literature have been established by deriving relations among the ranks of the coding matrices, whereas our converse proof in Section \ref{sec:optimal} is information-theoretic in nature. The benefit of our approach is that it allows us to derive converse bounds for the general class of codes in a single framework. 

\subsection{Performance Metrics}

The performance of an $N$-server PIR code can be measured using the following three quantities:
\begin{enumerate}[itemsep=0pt]
\item The retrieval rate
\begin{align}
R\triangleq \frac{L\log_2 |\mathcal{X}|}{\log_2|\mathcal{Y}|\sum_{n=0}^{N-1}\Expt(\ell_n)}, \label{rate_def}
\end{align}
which is the number of bits of desired message information that can be privately retrieved per bit of downloaded data. This quantity should be maximized, because higher rate implies fewer number of bits to be downloaded when retrieving a message. It was shown in \cite{sun2017PIRcapacity} that the retrieval rate is upper-bounded by the PIR capacity $C$, {\em i.e., } $R\leq C$, which is a function of $(N,K)$ as given in (\ref{eqn:capacityfomular}). 
\item The message size $L\log_2|\mathcal{X}|$, which is the number of bits to represent each individual message. This quantity should also be minimized, because PIR schemes for a larger message size can be constructed by concatenating multiple schemes for a smaller message size, but not vice versa. Therefore, in practical applications, a smaller message size implies a more versatile code design; similar considerations of reducing the sub-packetization factor also exist for the regenerating code problem, {\em e.g., }\cite{sasidharan2016explicit, YeBarg:17packetization,goparaju2014improved,balaji2017tight, li2018generic}, and the coded caching problem, {\em e.g., } \cite{shanmugam2016finite,yan2017placement,Tang:16}. Note that we refer to the parameter $L$ as the message length, while the definition of the message size also takes into account the alphabet size $|\mathcal{X}|$.
\item The upload cost 
\begin{align}
\sum_{n=0}^{N-1} \log_2  |\mathcal{Q}_n|,
\end{align}
which is the number of bits required to send the queries to the servers.  This quantity should be minimized for an efficient PIR code, since a smaller upload cost implies less user-to-server communication, and simpler server functions as mentioned earlier.
\end{enumerate}

The code construction we shall propose in this work is optimal in the following senses:
\begin{enumerate}[itemsep=0pt]
\item It is capacity-achieving $R=C$, {\em i.e.,} the retrieval rate is optimal;
\item It has the smallest, thus optimal, message size among all capacity-achieving uniformly decomposable codes;
\item It has the smallest, thus optimal, upload cost among all capacity-achieving decomposable codes.
\end{enumerate}

\section{A New Capacity-Achieving PIR Code}	
\label{sec:newPIR}

In this section, we provide the details of the proposed codes. Before presenting the code construction under general parameters, we provide a motivating example for the case of  $(N,K)=(2,2)$.

\subsection{A Motivating Example $(N,K)=(2,2)$}
\label{sec:motivatingexample}

\begin{table}
\centering
\caption{Answers for message $A$ and $B$ for $(N,K)=(2,2)$ in a simple code. \label{table:2_2alter}}
\vspace{0.2cm}
\begin{tabular}{|c||c|c||c|c|}\hline
\multirow{2}{*}{}&\multicolumn{2}{c||}{Requesting $A$}&\multicolumn{2}{c|}{Requesting $B$}\\\cline{2-5}
                            &Server-1&Server-2&Server-1&Server-2\\\hline\hline
$\mathsf{F}=0$&$0$  &$a$&$0$    &$b$\\\hline
$\mathsf{F}=1$&$a+b$&$b$&$a+b$&$a$\\\hline
\end{tabular}
\end{table}

Let us consider two different codes for the $(N,K)=(2,2)$ case.

\begin{itemize}
\item A very simple but new capacity-achieving code is as given in Table \ref{table:2_2alter}, where $0$ in the transmission means no symbol is transmitted. Here the two messages $A=(a)$ and $B=(b)$, each of which has only $1$ symbol. The random key is binary, uniformly distributed in the key set $\mathcal{F}=\{0,1\}$. It can be seen that the expected download cost is $0.5+1=1.5$, and thus the rate is $2/3$, which achieves the capacity. 
\item In comparison, in the code constructed in \cite{sun2017PIRcapacity}, the two messages $A=(a_1,a_2,a_3,a_4)$ and $B=(b_1,b_2,b_3,b_4)$ each have $4$ symbols. The random key set $\mathcal{F}$ is the collection of permutation $\pi(\cdot)$, which is used to select the one-to-one correspondence between $\{1,2,3,4\}$ and $\{\Box,\Diamond,\clubsuit,\heartsuit\}$. With this correspondence determined, the code is as given in Table \ref{table:2_2}. The download cost is $6$ symbols, and the rate is thus $2/3$. 
\end{itemize}


\begin{table}
\centering
\caption{Answers for message $A$ and $B$ for $(N,K)=(2,2)$ in \cite{sun2017PIRcapacity}. \label{table:2_2}}
\vspace{0.2cm}
\begin{tabular}{|c|c||c|c|}\hline
\multicolumn{2}{|c||}{Requesting $A$}&\multicolumn{2}{c|}{Requesting $B$}\\\hline
Server-1&Server-2&Server-1&Server-2\\\hline\hline
$a_{\Box},b_{\Box},a_\clubsuit+b_\Diamond$&$a_\Diamond,b_\Diamond,a_\heartsuit+b_{\Box}$&$a_{\Box},b_{\Box},a_\Diamond+b_\clubsuit$&$a_\Diamond,b_\Diamond,a_{\Box}+b_\heartsuit$\\\hline
\end{tabular}
\end{table}

%


It is observed that in the simple new code, the message sizes for different queries are allowed to vary, while the code constructed in \cite{sun2017PIRcapacity} uses answers of the same length, regardless of the key realization and the query. 

\subsection{The New PIR Code}
\label{sec:newcode}
The code we propose, which will be referred to as the $N$-ary-indexed PIR code, has the following parameter:
\begin{gather}
L=N-1.
\end{gather}
The query sets at the servers are defined as
\begin{align}
&\mathcal{Q}_n\triangleq \left\{q_{n,0:K-1}\triangleq(q_{n,0},q_{n,1},\ldots,q_{n,K-1})\in \{0,1,\ldots,N-1\}^K\bigg{|} \left(\sum_{k=0}^{K-1} q_{n,k}\right)_N = n \right\},\nonumber\\
&\qquad\qquad\qquad\qquad\qquad\qquad\qquad\qquad\qquad\qquad\qquad\qquad\quad n\in \{0,1,\ldots,N-1\},
\end{align}
where $\left(\cdot\right)_N$ means the modulo $N$ operation. In other words, the queries are length-$K$ vectors, whose elements are in the set $\{0,1,\ldots,N-1\}$; the query set for server-$n$ is all such vectors whose elements sum up to $n$ under modulo $N$. 
It is easy to see that 
\begin{align}
|\mathcal{Q}_n|=N^{K-1},\quad n\in \{0,1,\ldots,N-1\},
\end{align}
since the first $K-1$ digits of the query, {\em i.e.}, $(q_{n,0},q_{n,1},\ldots,q_{n,K-2})$, can take any value in the set $\{0,1,\ldots,N-1\}^{K-1}$, however, for a fixed server-$n$, the last digit is then uniquely determined in the set $\mathcal{Q}_n$.

The sample space of the random key is defined as $\mathcal{F}=\{0,1,\ldots,N-1\}^{K-1}$, and thus the random key $\mathsf{F}$ can be written as
\begin{align}
\mathsf{F}=(\mathsf{F}_0,\mathsf{F}_1,\ldots,\mathsf{F}_{K-2}),
\end{align}
where $\mathsf{F}_k\in \{0,1,\ldots,N-1\}$, $k=0,1,\ldots,K-2$. 
Each message $W_k$, $k\in \{0,1,\ldots,K-1\}$, is a length-$L$ vector and thus by pre-pending a dummy variable $W_{k,0}\triangleq 0$, can be written as
\begin{align}
W_k=(W_{k,0},W_{k,1},\ldots,W_{k,N-1}),
\end{align}
where  $(W_{k,1},\ldots,W_{k,N-1})$ is the true information payload of the message $W_k$.
Without loss of generality, we shall assume $\mathcal{X}=\{0,1,\cdots,|\mathcal{X}|-1\}$, which, together with the modulo addition operation $\oplus$, forms a finite group $(\mathcal{X},\oplus)$.
This includes the particularly attractive choice of $\mathcal{X}=\{0,1\}$,  where each symbol is a bit and the group addition is simply binary XOR, and in this case, the binary group can also be viewed as the binary field. 

We next provide the precise forms of the four coding functions with the parameter and the relevant sets defined above, which constitute the proposed code:
\begin{enumerate}
\item The query function $\phi_n$ for $n\in \{0,1,\ldots,N-1\}$ is 
\begin{gather}
Q^{[k]}_n=\phi_n(k,\mathsf{F})=\left(\mathsf{F}_0,\mathsf{F}_1,\mathsf{F}_{k-1},\left(n-\mathsf{F}^*_{k}\right)_N,\mathsf{F}_{k},\ldots,\mathsf{F}_{K-2}\right), 
\end{gather}
where $\mathsf{F}^*_{k}\triangleq \left(\sum_{i=0}^{K-2}\mathsf{F}_i\right)_N $. In other words, all digits except the $k$-th digit in the query vector are copied from $\mathsf{F}$, while the $k$-th digit is set to match the unique value in the query set at this server. This query can be equivalently written as $Q^{[k]}_{n,0:K-1}$ since it is a length-$K$ vector.
\item The answer length function $\ell_n$ for $n\in \{0,1,\ldots,N-1\}$ is 
\begin{gather}
 \ell_n(n,q)=\left\{
\begin{array}{ll}
0 &(n,q)=(0,(0,0,\ldots,0))\\
1 &\text{otherwise}
\end{array}
\right..
\end{gather}
In other words, there is only one query at the $0$-th server that will induce $\ell_0=0$, while all other queries at all other servers will induce an answer of a single symbol.
\item The answer function $\varphi_n$ for $n\in \{0,1,\ldots,N-1\}$ is 
\begin{align}
&A^{[k]}_n=\varphi_n(Q^{[k]}_{n,0:K-1},W_{0:K-1})\nonumber\\
&=W_{0,Q^{[k]}_{n,0}}\oplus W_{1,Q^{[k]}_{n,1}}\oplus\ldots\oplus W_{K-1,Q^{[k]}_{n,K-1}}\nonumber\\
&=W_{k,\left(n-\mathsf{F}^*_{k}\right)_N}\oplus \left(W_{0,\mathsf{F}_0}\oplus\ldots \oplus W_{k-1,\mathsf{F}_{k-1}}\oplus W_{k+1,\mathsf{F}_k}\oplus\ldots\oplus W_{K-1,\mathsf{F}_{K-2}}\right),
\end{align}
where $\oplus$ is the addition operation in the group $\mathcal{X}$. 
For conciseness, we shall define
\begin{align}
F\triangleq W_{0,\mathsf{F}_0}\oplus\ldots \oplus W_{k-1,\mathsf{F}_{k-1}}\oplus W_{k+1,\mathsf{F}_k}\oplus\ldots\oplus W_{K-1,\mathsf{F}_{K-2}}.
\end{align}
\item The answers from the servers are
\begin{align}
A^{[k]}_n=W_{k,\left(n-\mathsf{F}^*_{k}\right)_N}\oplus F,\quad n\in \{0,1,\ldots,N-1\}.
\end{align}
The message $W_k$ can now be reconstructed by computing
\begin{align}
W_{k,\left(n-\mathsf{F}^*_{k}\right)_N}=A^{[k]}_n\ominus A^{[k]}_{\mathsf{F}^*_{k}}=A^{[k]}_n\ominus F,\quad n\in \{0,1,\ldots,N-1\},\label{eqn:allW}
\end{align}
where $\ominus$ is the subtraction operation in the abelian group $\mathcal{X}$.
\end{enumerate}

The correctness of this code is almost immediate, once we observe that in (\ref{eqn:allW}), as $n$ ranges in the set $\{0,1,\ldots,N-1\}$, the corresponding value $\left(n-\mathsf{F}^*_{k}\right)_N$ exhausts all possible values in $\{0,1,\ldots,N-1\}$ as well. This implies all the elements $W_{k,n}$, $n\in \{0,1,\ldots,N-1\}$ are recovered, and thus the message is correctly reconstructed. 
The privacy of the code is also almost immediate, as for any $k\in \{0,1,\ldots,K-1\}$, $n\in \{0,1,\ldots,N-1\}$, and $q\in\mathcal{Q}_n$,
\begin{align}
\mathbf{Pr}(Q^{[k]}_n=q)=N^{-K+1},
\end{align}
{\em i.e.}, the queries are sent to a server with a uniform distribution on the respective query set. Since at each server, each answer is sent with probability $N^{-K+1}$, and only one answer in server-$0$ has length $0$ while all other answers have length $1$, 
the rate of the code is
\begin{align}
R=\frac{N-1}{(1-N^{-K+1})+(N-1)}=\frac{N-1}{N-N^{-K+1}}=\left(1+\frac{1}{N}+\frac{1}{N^2}+\ldots+\frac{1}{N^{K-1}}\right)^{-1}=C,
\end{align}
{\em i.e.}, achieving the capacity. The upload cost is simply given by 
\begin{align}
N\log_2N^{K-1}=N(K-1)\log_2 N,
\end{align}
which is roughly linear in $K$ for any fixed $N$. 

We summarize the properties of the proposed PIR code construction in the following theorem. 
\begin{theorem}
The $N$-ary-indexed PIR code is correct, privacy-preserving, and capacity-achieving. Among all capacity-achieving uniformly decomposable PIR codes, it has the smallest message size, which is $N-1$. Among all capacity-achieving decomposable PIR  codes, it has the lowest upload cost, which is $N(K-1)\log_2 N$. 
\end{theorem}

The optimality in terms of the message size and the upload cost is proved in Section \ref{sec:optimal}. 
The capacity-achieving code in \cite{sun2017PIRcapacity} has a message size of $L=N^K$ and an upload cost of $NK \log_2(\frac{N^K!}{N^{K-1}!})$, while the one in \cite{sun2017optimal} has a message size\footnote{The definition of retrieval rate (\ref{rate_def}) is given in terms of the inverse of the expected number of downloaded symbols (over all random queries), which is in line with the approach taken in \cite{sun2017PIRcapacity}. In \cite{sun2017optimal}, an alternative definition was adopted, where the retrieval rate was defined in terms of the inverse of the maximum number of downloaded symbols (among all possible queries). Under the alternative definition of \cite{sun2017optimal}, the minimum message size was shown to be $N^{K-1}$ for any capacity-achieving codes. In a sense, our result shows that this subtle difference in the problem definition in fact induces a significant difference in terms of the optimal message sizes.} of $L = N^{K-1}$ and an upload cost of $NK \log_2(\frac{N^{K-1}!}{N^{K-2}!})$. Therefore, the proposed code construction is able to provide an exponential order of improvements over the existing ones in the literature.

\subsection{An Example for $(N,K)=(3,3)$}
\label{sec:example}

Here we use $(N,K)=(3,3)$ to illustrate the general code construction. The code will have $L=N-1=2$, and we shall denote $W_0=(a_1, a_2), W_1=(b_1, b_2), W_2=(c_1, c_2)$, where all the elements are in the binary field $\{0,1\}$. As described in the general code construction, we extend these messages by pre-pending one dummy element to each of them, denoted as $a_0=b_0=c_0=0$, to form 
\begin{align}
W_0=(a_0,a_1, a_2), W_1=(b_0,b_1, b_2), W_2=(c_0,c_1, c_2).
\end{align}
In Table \ref{table:3_3}, we provide the query set $\mathcal{Q}_n$ at each server, as well as the corresponding answers.

Let us consider the case where the random key is chosen to be $\mathsf{F}=(0,2)$, and the message being requested is $W_1$, then the three queries sent to the servers are
\begin{align}
q_0=(0,1,2),\quad q_1=(0,2,2),\quad q_2=(0,0,2),
\end{align}
{\em i.e.}, the middle digit in the query is chosen to be the unique value in each query set, and the other two digits are set according to $\mathsf{F}=(0,2)$. The answers are thus
\begin{align}
A_0=a_0\oplus b_1 \oplus c_2=b_1 \oplus c_2,\quad A_1=a_0\oplus b_2 \oplus c_2=b_2 \oplus c_2,\quad A_2=a_0\oplus b_0 \oplus c_2=c_2.
\end{align}
It is clear that $b_1$ and $b_2$ can be recovered from these answers by subtracting $A_2=c_2$ from $A_0$ and $A_1$. The code is also privacy-preserving, since regardless of the message being requested, a query element is being sent with probability $1/9$. The retrieval rate is also easy to compute as
\begin{align}
R=\frac{2}{\frac{1}{9}*2+\frac{8}{9}*3}=\frac{9}{13},
\end{align}
which matches the capacity of this system.

\vspace{0.2cm}
\noindent\textbf{Remark: } The queries in each row of Table \ref{table:3_3} are intentionally arranged to have the first two digits being the same, for ease of inspection.

\begin{table*}[tb!]
\begin{center}
\caption{The query sets and the answers at the servers. \label{table:3_3}}
\vspace{0.2cm}
\begin{tabular}{|c|c||c|c||c|c|}
\hline
\multicolumn{2}{|c||}{Server-0}&\multicolumn{2}{c||}{Server-1}&\multicolumn{2}{c|}{Server-2}\\\hline
${Q}_0$& answers&${Q}_1$& answers&$Q_2$& answers\\\hline
000& $0$                                    &001&$a_0\oplus b_0\oplus c_1$     &002&$a_0\oplus b_0\oplus c_2$\\
012& $a_0\oplus b_1\oplus c_2$   &010&$a_0\oplus b_1\oplus c_0$     &011&$a_0\oplus b_1\oplus c_1$\\
021& $a_0\oplus b_2\oplus c_1$   &022&$a_0\oplus b_2\oplus c_2$     &020&$a_0\oplus b_2\oplus c_0$\\
102& $a_1\oplus b_0\oplus c_2$   &100&$a_1\oplus b_0\oplus c_0$     &101&$a_1\oplus b_0\oplus c_1$\\
111& $a_1\oplus b_1\oplus c_1$   &112&$a_1\oplus b_1\oplus c_2$     &110&$a_1\oplus b_1\oplus c_0$\\
120& $a_1\oplus b_2\oplus c_0$   &121&$a_1\oplus b_2\oplus c_1$     &122&$a_1\oplus b_2\oplus c_2$\\
201& $a_2\oplus b_0\oplus c_1$   &202&$a_2\oplus b_0\oplus c_2$     &200&$a_2\oplus b_0\oplus c_0$\\
210& $a_2\oplus b_1\oplus c_0$   &211&$a_2\oplus b_1\oplus c_1$     &212&$a_2\oplus b_1\oplus c_2$\\
222& $a_2\oplus b_2\oplus c_2$   &220&$a_2\oplus b_2\oplus c_0$     &221&$a_2\oplus b_2\oplus c_1$\\
\hline
\end{tabular}
\end{center}
\end{table*}

\section{Lower Bounding the Message Size and the Upload Cost}
\label{sec:optimal}

The minimum upload cost and the minimum message size are closely related to the retrieval rate of a PIR code. For example, a naive PIR code where everything is downloaded can have upload cost of 0, and message size of 1, however a more efficient PIR code will need to induce a larger message size and a higher upload cost. In this work, we consider the minimum upload cost and the minimum message size when the retrieval rate is maximized and when the codes are decomposable, {\em i.e.}, capacity-achieving decomposable codes. We will show, through a delicate set of relations among the coding function matrices $G_n^{(q)}$'s, that the capacity-achieving requirement forces the PIR codes to have certain algebraic structure, which can be utilized to derive the desired lower bounds. 

\subsection{Properties of Capacity-Achieving Decomposable Codes}

We first provide a detailed analysis of capacity-achieving codes, from which three important properties are derived, given in two lemmas. The analysis is a refinement of the converse proof given in \cite{sun2017PIRcapacity}, however, with the emphasis on the necessary conditions for optimal codes. A similar approach was used in \cite{tian2017matched} to analyze optimal joint source-channel codes, and in \cite{Tian:16Computer} to facilitate reverse-engineering code designs. 
\begin{lemma}
\label{lemma:first}
For any PIR code, we have 
\begin{align}
I\left(W_{0:k-1,k+1:K-1};A^{[k]}_{0:N-1}\Big{|}W_k,\mathsf{F}\right)\leq L(1/R-1)\log_2\left|\mathcal{X}\right|,\quad k\in \{0,1,\ldots,K-1\}.\label{eqn:P1}
\end{align}
Moreover, for any PIR code that the equality holds for all $k\in \{0,1,\ldots,K-1\}$ in (\ref{eqn:P1}), let $q_{0:N-1}=(q_0,q_1,\ldots,q_{N-1})$ be a set of queries for which $\mathbf{Pr}(Q^{[k]}_{0:N-1}=q_{0:N-1})>0$ for some $k\in \{0,1,\ldots,K-1\}$, then the code must have 
\begin{enumerate}
\item[\textbf{P1}.]  Independence of the retrieved data: the $N$ random variables $A^{(q_0)}_{0},A^{(q_1)}_{1},\ldots,A^{(q_{N-1})}_{N-1}$ are mutually independent, where $A^{(q_n)}_{n}$ is the answer from server-$n$ when the query  $Q_n^{[k]}=q_n$.
\end{enumerate}
\end{lemma}

The proof of this lemma is given in the appendix. The property \textit{\textbf{P1}} is obtained by setting the inequality (\ref{eqn:P1}) to equality, which forces the intermediate steps to also become inequality, and then extracting the independence implied by such information theoretic equality. 

\vspace{0.3cm}
\noindent\textbf{Remark: } For decomposable codes, we can further write
\begin{align}
\left(A^{(q_0)}_{0},A^{(q_1)}_{1},\ldots,A^{(q_{N-1})}_{N-1}\right)=\left(W_{0:K-1}\cdot G^{(q_0)}_0,W_{0:K-1}\cdot G^{(q_1)}_1,\ldots,W_{0:K-1}\cdot G^{(q_{N-1})}_{N-1}\right).
\end{align}
Also note that for linear codes, the independence relation given above implies that the columns of the matrices $G^{(q_0)}_0,G^{(q_1)}_1,\ldots,G^{(q_{N-1})}_{N-1}$ are linearly independent. 
Recall that for decomposable codes, the answer for a query $Q_n=q$ at server-$n$ can be written as $W_{0:K-1}\cdot G^{(q)}_n$, or more concisely, sometimes represented by the coding function matrix $G^{(q)}_n$ alone. The next lemma involves submatrices of $G^{(q)}_n$, with the rows corresponding to a subset of the messages removed, say $\{W_i, i\in \mathcal{A}\}$; we shall write such a submatrix as $G^{(q)}_{n|\mathcal{A}}$. For example, if $(N,K)=(3,3)$, and $\mathcal{A}=1$, then $G^{(1)}_{1|1}$ is the submatrix of $G^{(1)}_{1}$ with the middle row corresponding to the message $W_1$ removed. 

\begin{lemma}
\label{lemma:second}
Let $\pi:\{0,1,\ldots,K-1\}\rightarrow\{0,1,\ldots,K-1\}$ be a permutation function. 
For any PIR code, for any $k\in \{1,2,\ldots,K-1\}$,
\begin{align}
NI\left(W_{\pi(k:K-1)};A_{0:N-1}^{[\pi(k-1)]}\Big{|}W_{\pi(0:k-1)},\mathsf{F}\right)\geq I\left(W_{\pi(k+1:K-1)};A_{0:N-1}^{[\pi(k)]}\Big{|}W_{\pi(0:k)},\mathsf{F}\right)+L\log_2|\mathcal{X}|.\label{eqn:P2}
\end{align}
Moreover, for any decomposable code for which the equality holds for any $k$ and $\pi(\cdot)$ in (\ref{eqn:P2}), let $q_{0:N-1}=(q_0,q_1,\ldots,q_{N-1})$ be a set of queries for which $\mathbf{Pr}(Q^{[k]}_{0:N-1}=q_{0:N-1})>0$ for the query of the message $W_k$, and $G^{(q_0)}_{0},G^{(q_1)}_{1},\ldots,G^{(q_{N-1})}_{N-1}$ be the corresponding answer coding matrices, then
\begin{enumerate} 
\item[\textbf{P2}.] Identical information for the residuals: the $N$ random variables 
$$W_{0:k-1,k+1:K-1}\cdot G^{(q_0)}_{0|k}, W_{0:k-1,k+1:K-1}\cdot G^{(q_1)}_{1|k}, ..., W_{0:k-1,k+1:K-1}\cdot G^{(q_{N-1})}_{N-1|k}$$ are deterministic of each other;
\item[\textbf{P3}.] Independence of the requested message signals: the random variables 
$$W_k\cdot G^{(q_0)}_{0|0:k-1,k+1:K-1},W_k\cdot G^{(q_1)}_{1|0:k-1,k+1:K-1},\ldots,W_k\cdot G^{(q_{N-1})}_{N-1|0:k-1,k+1:K-1}$$
are independent.
\end{enumerate}
\end{lemma}

The proof of this lemma can be found in the appendix. 

\noindent\textbf{Remark: } The property of decomposable codes was used in the proof of Lemma \ref{lemma:second}, where the answers are decomposed into separate components according to the messages $W_k$'s, with which relations among these answers are derived. Such decomposition does not apply on other code classes in general, and thus the proof cannot be carried through using the same argument.    

\begin{theorem}
\label{theorem:minmessage}
Any capacity-achieving decomposable PIR code must have the properties \textit{\textbf{P1}}-\textit{\textbf{P3}}. 
\end{theorem}
\begin{proof}
Let $\pi: \{0,1,\ldots,K-1\}\rightarrow\{0,1,\ldots,K-1\}$ be a permutation. Starting from Lemma \ref{lemma:first}, we can write
\begin{align}
\frac{L}{R}-L&\geq I\left(W_{\pi(1:K-1)};A_{0:N-1}^{[\pi(0)]}\Big{|}W_{\pi(0)},\mathsf{F}\right)\nonumber\\
&\geq \frac{L}{N}+\frac{1}{N}I\left(W_{\pi(2:K-1)};A_{0:N-1}^{[\pi(1)]}\Big{|}W_{\pi(0:1)},\mathsf{F}\right)\nonumber\\
&\geq ...\nonumber\\
&\geq L\left(\frac{1}{N}+\ldots+\frac{1}{N^{K-1}}\right),
\end{align}
where all the other inequalities are by recursively applying Lemma \ref{lemma:second}, and it follows that $R\leq C$. For any decomposable code that satisfies $R=C$, all the inequalities in Lemma \ref{lemma:first} and Lemma \ref{lemma:second} must be equality for any permutation $\pi$, and according to the lemmas, such decomposable codes must have properties \textit{\textbf{P1}}-\textit{\textbf{P3}}.
\end{proof}
 
\subsection{Minimum Message Size}

We have the following theorem, which provides a lower bound on the minimum message size for capacity-achieving uniformly decomposable codes.

\begin{theorem}
The message size of any uniformly decomposable capacity-achieving PIR code is greater than or equal to $(N-1)\log_2|\mathcal{Y}|$; in particular, it must be greater than or equal to $(N-1)$.
\end{theorem}

\vspace{0.2cm}
\noindent\textbf{Remark:} Clearly this implies that the standard linear codes defined on finite fields are lower bounded by the same values. Note also that the bound $(N-1)\log_2|\mathcal{Y}|$ is dependent on $\mathcal{Y}$ but not $\mathcal{X}$, which reflects the fact that the representation of the message is of little fundamental importance because we can always use an equivalent representation.  

\begin{proof}
Let us consider a capacity-achieving uniformly decomposable PIR code, and the request to retrieve the message $W_k$. Recall property \textit{\textbf{P2}} which states that $W_{0:k-1,k+1:K-1}\cdot G^{(q_n)}_{n|k}$, $n\in\{0,1,\ldots,N-1\}$, are deterministic functions of each other. There must be a set of queries $q_{0:N-1}$ with non-zero probability such that
\begin{align}
H\left(W_{0:k-1,k+1:K-1}\cdot G^{(q_n)}_{n|k}\right)\neq 0,\quad n=\{0,1,\ldots,N-1\},\label{eqn:nonzero}
\end{align}
because otherwise, all answers will have the form
\begin{align}
W_{0:K-1}\cdot G^{(q_n)}_{n}=&W_{k}\cdot G^{(q_n)}_{n|0:k-1,k+1:K+1}\oplus \Delta,\quad n=\{0,1,\ldots,N-1\},
\end{align}
where $\Delta\in\mathcal{Y}$ is a constant; this would imply that the answers only involve the message $W_k$ but not other messages, but such answers clearly cannot be both private and correct. 

With such a set of queries $q_{0:N-1}$ that (\ref{eqn:nonzero}) holds, consider property \textit{\textbf{P3}}, which states that 
\begin{align}
W_k\cdot G^{(q_0)}_{0|0:k-1,k+1:K-1},W_k\cdot G^{(q_1)}_{1|0:k-1,k+1:K-1},\ldots,W_k\cdot G^{(q_{N-1})}_{N-1|0:k-1,k+1:K-1},\label{eqn:Wkterms}
\end{align}
are independent, and our aim is to show that no more than one of their entropies can be zero. To see this, assume otherwise, {\em i.e., } at least two of the entropies are zero. Without loss of generality, let us assume that 
\begin{align}
H\left(W_k\cdot G^{(q_0)}_{0|0:k-1,k+1:K-1}\right)=H\left(W_k\cdot G^{(q_1)}_{1|0:k-1,k+1:K-1}\right)=0,
\end{align}
implying that both $W_k\cdot G^{(q_0)}_{0|0:k-1,k+1:K-1}$ and $W_k\cdot G^{(q_1)}_{1|0:k-1,k+1:K-1}$ in fact take a fixed value, independent of the value of $W_k$.  However, this further implies that the retrieved messages from server-0 and server-1 are
\begin{align}
W_{0:K-1}\cdot G^{(q_0)}_{0}=&W_{0:k-1,k+1:K-1}\cdot G^{(q_0)}_{0|k}\oplus\Delta_1,\nonumber\\
W_{0:K-1}\cdot G^{(q_1)}_{1}=&W_{0:k-1,k+1:K-1}\cdot G^{(q_1)}_{1|k}\oplus\Delta_2, \label{eqn:twoterms}
\end{align}
where $\Delta_1$ and $\Delta_2$ are two constants in the abelian group $\mathcal{Y}$. Because of property \textit{\textbf{P2}}, the two random variables in (\ref{eqn:twoterms}) are in fact deterministic of each other. However this contradicts property \textit{\textbf{P1}} which states that the retrieved contents are independent (recall that their entropies are not zero). Thus we can conclude that at least $N-1$ of the entropies of the terms in (\ref{eqn:Wkterms}) are not zero. Because the function $\varphi^{(q)}(n,i,k)$ induces a uniform probability distribution on the coded symbol alphabet $\mathcal{Y}$, and moreover, by the independence property of \textit{\textbf{P3}}, we can now conclude that the message size must be greater than or equal to $(N-1)\log_2|\mathcal{Y}|$. Since any meaningful alphabet $\mathcal{Y}$ must satisfy $|\mathcal{Y}|\geq 2$, the message size must be greater than or equal to $N-1$.  The proof is thus complete.
\end{proof}

\noindent\textbf{Remark: } The property of uniformly decomposable codes is only invoked during the proof in the last step, which requires the component functions to induce a uniform distribution on the coded alphabet.

\subsection{Minimum Upload Cost}
\begin{theorem}
\label{theorem:upload}
The upload cost of any capacity-achieving decomposable PIR code is greater than or equal to $N(K-1)\log N$. 
\end{theorem}

We need the following notion of distinctness in the proof.
\begin{definition}
Two random variables $A$ and $B$ are called information-theoretically distinct, or simply distinct, if $I(A;B)<\max(H(A),H(B))$.
\end{definition}

According to this definition, if a random variable can be obtained from another through an invertible transformation, they are not information-theoretically distinct.

\begin{proof}

We prove that for a capacity-achieving decomposable PIR code for $N$ servers and $K$ messages, the minimum upload cost to each server is at least $(K-1)\log_2 N$, {\em i.e., } it is a lower bound on $\log_2|\mathcal{Q}_n|$, $n\in \{0,1,\ldots,N-1\}$. To begin the proof, we find a set of queries $q_{0:N-1}$ for the message $W_0$, and assume that  the answers  have the property that the interference signal ({\em i.e., } the part of the answer that is not the requested message) is not null, {\em i.e., }
\begin{align}
H\left(W_{1:K-1}\cdot G^{(q_n)}_{n|0}\right)\neq 0,\quad n=\{0,1,\ldots,N-1\},
\end{align}
which always exists using the same argument as in Theorem \ref{theorem:minmessage}; {\em c.f.} (\ref{eqn:nonzero}). This implies that at least for one of the interference signals $k\in \{1,2,\ldots,K-1\}$ we have
\begin{align}
H\left(W_{k}\cdot G^{(q_n)}_{n|0:k-1,k+1:K-1}\right)\neq 0,\quad n=\{0,1,\ldots,N-1\},\label{eqn:interference0}
\end{align}
due to property \textit{\textbf{P2}}. Without loss of generality, let us assume it is $k=K-1$. With this set of queries, following the argument in Theorem \ref{theorem:minmessage}, at most one of the entropies
\begin{align}
H\left(W_{0}\cdot G^{(q_n)}_{n|1:K-1}\right),\quad n=\{0,1,\ldots,N-1\}, 
\end{align}
can be zero. Again without loss of generality, assume it is $n=0$. We shall denote a particular answer from a server as $V_{a_0,a_1,\ldots,a_{K-1}}$, the meaning of which will soon become apparent.

Since the queries $q_{0:N-1}$ are for the message $W_0$, the answers from the $N$ servers, respectively,
\begin{align}
V_{0,0,0,\ldots,0}, V_{1,0,0,\ldots,0}, \ldots,V_{N-1,0,0,\ldots,0},\label{eqn:Vs}
\end{align}
can be used to recover $W_0$, and moreover, the $W_0$ component functions
\begin{align}
W_{0}\cdot G^{(q_n)}_{n|1:K-1},\quad n=\{0,1,\ldots,N-1\}, \label{eqn:W0}
\end{align}
are all information-theoretically distinct by property \textit{\textbf{P3}}, and the fact that at most one of them can have zero entropy. We indicate this distinctness by the subscript in the answers (\ref{eqn:Vs}) in the $0$-th position.

Due to the privacy constraint, each answer in (\ref{eqn:Vs}) can also be used to reconstruct $W_1$, together with some other answers, {\em i.e., }
\begin{align}
\begin{array}{llllll}
\text{Server 0}&\text{Server 1}&\text{Server 2}&...&\text{Server $N-1$} \\\hline\hline
\underline{V_{0,0,0,\ldots,0}},&V_{0,1,0,\ldots,0},&V_{0,2,0,\ldots,0},&\ldots &V_{0,N-1,0,\ldots,0}&\rightarrow W_1\\
V_{1,N-1,0,\ldots,0},&\underline{V_{1,0,0,\ldots,0}},&V_{1,1,0,\ldots,0},&\ldots&V_{1,N-2,0,\ldots,0}&\rightarrow W_1\\
\ldots &\ldots &\ldots &\ldots &\ldots &\ldots\\
V_{N-1,1,0,\ldots,0},&V_{N-1,2,0,\ldots,0},&V_{N-1,3,0,\ldots,0},&\ldots&\underline{V_{N-1,0,0,\ldots,0}}&\rightarrow W_1,
\end{array}\label{eqn:multirows}
\end{align}
where in each row, the results produced by the component functions on $W_{0,2:N-1}$ in the answers are deterministic functions of each other across different servers, due to property \textit{\textbf{P2}}.
Moreover, for these answers, the component $W_{K-1}$ must satisfy
\begin{align}
H\left(W_{K-1}\cdot G^{(q_n)}_{n|0:K-2}\right)\neq 0,\quad n=\{0,1,\ldots,N-1\},\label{eqn:WN1}
\end{align}
due to our assumption on (\ref{eqn:interference0}) holding for $k=K-1$. As a consequence, in each row of (\ref{eqn:multirows}), the component functions $W_1\cdot G^{(q)}_{n|0,2:K-1}$ are again distinct. Note however, across rows of  (\ref{eqn:multirows}), the component functions on $W_1$ in the answers are not necessarily distinct or identical.  However, the component functions $W_{0:1}\cdot G^{(q)}_{n|2:K-1}$ of the answers in (\ref{eqn:multirows}) are all distinct, since they have distinct $W_0$ component functions in different rows , ({\em i.e.,} for answers in different rows, the $W_0$ component functions are (\ref{eqn:W0})), while for answers in the same row, the $W_1$ component functions are distinct. Thus there are at least $N$ answers with distinct component functions $W_{0:1}\cdot G^{(q)}_{n|2:K-1}$ at server-$n$, which are the answers with the sum of the indices equal to $n$ modulo $N$, given in the same column in (\ref{eqn:multirows}). 

Next consider each answer in (\ref{eqn:multirows}), which can also be used to recover $W_2$ due to the privacy requirement. For example, if we focus on the answers $V_{N-1,2,0,\ldots,0}$ and $V_{1,1,0,\ldots,0}$, which are from server $1$ and server $2$, respectively, they can be used to recover $W_2$ with some other answers
\begin{align}
\begin{array}{llllll}
\text{Server 0}&\text{Server 1}&\text{Server 2}&...&\text{Server $N-1$} \\\hline\hline
\ldots &\ldots &\ldots &\ldots &\ldots &\ldots\\
V_{N-1,2,N-1,\ldots,0},&\underline{V_{N-1,2,0,\ldots,0}},&V_{N-1,2,1,\ldots,0},&\ldots&V_{N-1,2,N-2,\ldots,0}&\rightarrow W_2\\
\ldots &\ldots &\ldots &\ldots &\ldots &\ldots\\
V_{1,1,N-2,\ldots,0},&V_{1,1,N-1,\ldots,0},&\underline{V_{1,1,0,\ldots,0}},&\ldots&V_{1,1,N-3,\ldots,0}&\rightarrow W_2\\
\ldots &\ldots &\ldots &\ldots &\ldots &\ldots.\\
\end{array}
\end{align}
Again these answers are distinct through a similar argument as before. Using this argument on all the answers in (\ref{eqn:multirows}) for the retrieval of $W_2$, it can seen that across all the servers, there are at least $N^3$ answers, whose component functions $W_{0:2}\cdot G^{(q)}_{n|3:K-1}$ are all distinct, and each server has at least $N^2$ answers whose corresponding component functions are distinct. We can continue this line of argument for messages $W_3,W_4,\ldots,W_{K-2}$, resulting in a total of $N^{K-1}$ answers at all the servers ($N^{K-2}$ at each server) in the form of 
\begin{align}
V_{a_0,a_1,\ldots,a_{K-2},0},\quad a_k\in\{0,1,\ldots,N-1\},\label{eqn:distinct}
\end{align}
whose component functions $W_{0:K-2}\cdot G^{(q)}_{n|K-1}$ are distinct. 

Next consider the reconstruction of the message $W_{K-1}$, for which we need to be more careful. In this case, we cannot assume the interference signals in the retrieval are not null, because $W_{K-1}$ is now the requested message, and the condition (\ref{eqn:WN1}) becomes insufficient; thus the component functions of $W_{K-1}$ in the answers cannot be guaranteed to be all distinct during a retrieval. However, notice that due to the distinctness of the component functions $W_{0:K-2}\cdot G^{(q)}_{n|K-1}$ in all the answers in (\ref{eqn:distinct}), at most one of these component functions can have zero entropy, {\em i.e., } in these answers, there is at most one of them satisfying,
\begin{align}
H\left(W_{0:K-2}\cdot G^{(q)}_{n|K-1}\right)=0.\label{eqn:zeroentropy}
\end{align} 
For all other answers that (\ref{eqn:zeroentropy}) does not hold, our previous induction argument based on the distinctness of the signal components still applies. For the one exception answer where (\ref{eqn:zeroentropy}) holds, which we assume without of generality to be $V_{0,0,\ldots,0}$, the answers to recover $W_{K-1}$ can be labeled as
\begin{align}
V_{0,0,\ldots,0}, V_{0,0,\ldots,1},\ldots,V_{0,0,\ldots,N-1},
\end{align}
which may not be all distinct since there may be more than one item with zero entropy. However, since they are placed at different servers, each one of them is distinct from all other answers at the same server. 

We have shown that at server-$n$, $n\in\{0,1,\ldots,N-1\}$, there are at least $N^{K-1}$ distinct answers $V_{a_0,a_1,\ldots,a_{K-1}}$ in the form of 
\begin{align}
\left(\sum_{k=0}^{K-1}a_k\right)_N=n,
\end{align} 
implying $|\mathcal{Q}_n|\geq N^{K-1}$. Our proof is now complete.
\end{proof}

\noindent\textbf{Remark: } Although Theorem \ref{theorem:upload} is stated in terms of the total upload cost, in the proof, we have actually shown that the upload cost at each individual server is greater than or equal to $(K-1)\log N$. 
 
\section{Symmetry and Symmetrized Codes}

\label{sec:symmetry}

The proposed code construction is able to achieve exponential improvements over the existing capacity-achieving PIR codes in the literature, in terms of both the message size and the upload cost. The question we wish to address in this section is what the root cause is for these improvements. It is clear that the existing codes in the literature, such as  \cite{sun2017PIRcapacity, sun2017optimal,banawan2018capacity, Wang_Skoglund, Tandon_CachePIR,Banawan_Ulukus_Multimessage},  are all symmetric, while our proposed code is not symmetric. It is thus natural to suspect that this symmetry vs. asymmetry relation is the root cause, however, in order to better understand this issue, we have to identify and evaluate carefully the symmetry relations in the problem. It should be noted that the symmetrization techniques given this section should not be viewed as design requirements stipulated by practical system design considerations, but rather should be viewed as theoretical tools to  pinpoint the key difference between our code construction and the existing ones, and perhaps to help future investigations on the capacities of privacy-preserving primitives, as they appear to be rather general. 

Recall our discussion on the minimum upload cost, which is related to $|\mathcal{Q}_n|$. For simplicity, we shall refer to the distinct answers (or precisely, distinct answer functions) at a server as the varieties of the answers at this server\footnote{The term variety here should be distinguished from the algebraic variety concept in algebraic geometry.}. This concept plays an instrumental role in the subsequent discussion.  

There are in fact three kinds of symmetry relations in this problem setting:
\begin{enumerate}[itemsep=0pt]
\item Server-symmetry: obtained by permuting the servers;
\item Message-symmetry: obtained by permuting the messages;
\item Variety-symmetry: obtained by compositing the varieties of answers. 
\end{enumerate}
Among the three types of symmetry relations, the variety-symmetry is the most interesting, and appears unique to the PIR problem. Through this symmetry, it can be shown that without loss of optimality on the retrieval rate, we can always assume that the varieties are requested with a uniform distribution at any given server.  
These three symmetry components can be operated in composition, and space sharing of all possible permuted codes eventually can yield a highly symmetric code. In this section we shall provide a precise characterization of these three types of symmetry relations, and discuss several consequences of these relations. Technically, this is accomplished by providing a new set of coding functions, which by space-sharing over some permutations will induce certain symmetry relation on the coding rates and the probability distribution.

Central to these symmetry relations are the following random variables
\begin{align}
&\{W_0,W_1,\ldots,W_{K-1}\},\{A^{(0)}_{0},A^{(1)}_{0},\ldots,A^{(|\mathcal{Q}_0|-1)}_0\},\{A^{(0)}_{1},A^{(1)}_{1},\ldots,A^{(|\mathcal{Q}_1|-1)}_1\},\ldots,\nonumber\\
&\qquad\qquad\qquad \qquad\qquad\qquad\qquad \qquad\qquad\qquad\ldots, \{A^{(0)}_{N-1},A^{(1)}_{N-1},\ldots,A^{(|\mathcal{Q}_{N-1}|-1)}_{N-1}\},
\end{align}
where $A^{(q)}_n$ is the answer at server-$n$ for the query $Q_n=q$. Note that $A^{(q)}_n$ is a deterministic function of the messages $W_{0:K-1}$; this should be distinguished from $A^{[k]}_n$ which is the (randomized) answer for the request of the message $W_k$ at server-$n$, and it is not a deterministic function of the messages $W_{0:K-1}$.

\subsection{Server-symmetry}
Let $\pi(\cdot)$ be a permutation function on the set $\{0,1,\ldots,N-1\}$, which is the set of server indices. For any PIR code which is specified by the four coding functions in Definition \ref{def:problem},  a new set of coding functions can be specified as 
\begin{gather}
\hat{\phi}_n=\phi_{\pi(n)},\quad \hat{\ell}_n=\ell_{\pi(n)},\quad \hat{\varphi}_{n}=\varphi_{\pi(n)},\quad n\in\{0,1,\ldots,N-1\},\nonumber\\
\quad\hat{\psi}(A_{0:N-1},k,\mathsf{F})=\psi(A_{\pi^{-1}(0:N-1)},k,\mathsf{F}).
\end{gather}

Let us examine an example where $N=4$, and let
\begin{align}
\pi([0,1,2,3])=[3,0,1,2].
\end{align}
Then we have
\begin{align}
\hat{\phi}_0=\phi_{3},\quad \hat{\phi}_1=\phi_{0},\quad \hat{\phi}_2=\phi_{1},\quad \hat{\phi}_3=\phi_{2},
\end{align}
that is, the query sent to server-$0$ in this new code is what was sent to server-$3$, etc.. Similarly,
\begin{gather}
\hat{\ell}_{0}=\ell_{3},\quad \hat{\ell}_{1}=\ell_{0},\quad \hat{\ell}_{2}=\ell_{1},\quad \hat{\ell}_{3}=\ell_{2},\nonumber\\
\hat{\varphi}_{0}=\varphi_{3},\quad \hat{\varphi}_{1}=\varphi_{0},\quad \hat{\varphi}_{2}=\varphi_{1},\quad \hat{\varphi}_{3}=\varphi_{2},
\end{gather}
that is, the function to produce the answer (and the length of the answer) at server-$0$ in the permuted code is what was used at server-$3$ for the same query value, etc.; moreover, for the reconstruction function
\begin{align}
\hat{\psi}(A_0,A_1,A_2,A_3,k,\mathsf{F})=\psi(A_1,A_2,A_3,A_0,k,\mathsf{F}),
\end{align}
that is, the reconstructed message $\hat{W}_k$ using random key $\mathsf{F}$, is in fact obtained by operating the original function on the permuted answers, {\em i.e., } using the answer obtained from server-$0$ in the place of what was for the answer from server-$3$, etc..

It is easy to see that this new set of coding functions is indeed privacy-preserving and correct, since there is no essential change in the coding operations. 
A direct consequence of the definition of the new code is reflected on the equivalence of the induced random variables in the two codes
\begin{align}
&\{\hat{W}_0,\hat{W}_1,\ldots,\hat{W}_{K-1}\},\{\hat{A}^{(0)}_{0},\hat{A}^{(1)}_{0},\ldots,\hat{A}^{(|\hat{\mathcal{Q}}_0|-1)}_0\},\{\hat{A}^{(0)}_{1},\hat{A}^{(1)}_{1},\ldots,\hat{A}^{(|\hat{\mathcal{Q}}_1|-1)}_1\},\ldots,\nonumber\\
&\qquad\qquad\qquad\ldots,\{\hat{A}^{(0)}_{N-1},\hat{A}^{(1)}_{N-1},\ldots,\hat{A}^{(|\hat{\mathcal{Q}}_{N-1}|-1)}_{N-1}\}\nonumber\\
&=\{W_0,W_1,\ldots,W_{K-1}\},\{A^{(0)}_{\pi(0)},A^{(1)}_{\pi(0)},\ldots,A^{(|\mathcal{Q}_{\pi(0)}|-1)}_{\pi(0)}\},\{A^{(0)}_{\pi(1)},A^{(1)}_{\pi(1)},\ldots,A^{(|\mathcal{Q}_{\pi(1)}|-1)}_{\pi(1)}\},\ldots,\nonumber\\
&\qquad\qquad\qquad\ldots,\{A^{(0)}_{\pi(N-1)},A^{(1)}_{\pi(N-1)},\ldots,A^{(|\mathcal{Q}_{\pi(N-1)}|-1)}_{\pi(N-1)}\}.\label{eqn:equivalence}
\end{align}

Next consider the following code constructed through the space-sharing technique using a base code. Let each message consist of a total of $NL$ symbols, and apply a permuted version of the base code on each length-$L$ sequence (and over the $K$ messages), which corresponds to one of the cyclic permutations on $\{0,1,\ldots,N-1\}$. This space-sharing code is clearly privacy-preserving and correct, and it has the property that $|\hat{\mathcal{Q}}_0|=|\hat{\mathcal{Q}}_1|=\ldots=|\hat{\mathcal{Q}}_{N-1}|=\prod_{n=0}^{N-1} |\mathcal{Q}_n|$, {\em i.e., } the upload costs to all the servers are the same. Moreover, the expected retrieval rates are also the same across all the  servers, {\em i.e.,} $\Expt (\hat{\ell}_0)=\Expt(\hat{\ell}_1)=\ldots=\Expt(\hat{\ell}_{N-1})$.

We could also space share over longer messages of $(N!)L$ symbols each, where for each length-$L$ sequence we apply the permuted coding function corresponding to one of the $N!$ permutations on $\{0,1,\ldots,N-1\}$. By leveraging (\ref{eqn:equivalence}), it is also possible to obtain an invariance in terms of the joint entropy values of the subsets of the random variables. Such refined invariant relations are not necessary for this work, however, similar relations have been shown to be important when deriving information theoretic converse bounds \cite{Tian:JSAC13,Tian:16Computer} in other information systems.

It should be noted that although the expected numbers of retrieved symbols are the same across the servers (and thus the retrieval rates are the same per server), this does not imply for each individual set of queries $q_{0:N-1}$ with non-zero probability, the numbers of symbols being retrieved are the same as those for another set of queries $q'_{0:N-1}$. To achieve such a fine level of invariance, we will need to invoke the variety-symmetry, to be introduced in Section \ref{sec:variety}.   

\subsection{Message-symmetry}
Let $\pi(\cdot)$ be a permutation function on the set $\{0,1,\ldots,K-1\}$, which is the set of  message indices. For any PIR code which is specified by the four coding functions in Definition \ref{def:problem},  a new set of coding functions can be specified as
\begin{gather}
\bar{\phi}_n=\phi_{n},\quad \bar{\ell}_n=\ell_{n},\quad \bar{\varphi}_{n}(q,W_{0:K-1})=\varphi_{n}(q,W_{\pi(0:K-1)}),\quad n\in\{0,1,\ldots,N-1\},\nonumber\\
\quad\bar{\psi}(A_{0:N-1},k,\mathsf{F})=\psi(A_{0:N-1},\pi^{-1}(k),\mathsf{F}).
\end{gather}

Let us examine an example where $K=3$ and let
\begin{align}
\pi([0,1,2])=[2,0,1].
\end{align}
Then we have for the functions $\bar{\varphi}_{n}$
\begin{align}
\bar{\varphi}_n(q,W_0,W_1,W_2)=\varphi_n(q,W_2,W_0,W_1),\quad n\in \{0,1,\ldots,N-1\},
\end{align}
that is, the message $W_0$ in the new code serves the role of $W_1$ in the original code, etc..

For the reconstruction functions
\begin{gather}
\bar{\psi}(A_{0:N-1},0,\mathsf{F})=\psi(A_{0:N-1},1,\mathsf{F}),\quad \bar{\psi}(A_{0:N-1},1,\mathsf{F})=\psi(A_{0:N-1},2,\mathsf{F}),\nonumber\\
\bar{\psi}(A_{0:N-1},2,\mathsf{F})=\psi(A_{0:N-1},0,\mathsf{F}),
\end{gather}
that is, the message $W_0$ is reconstructed in the same way as that for $W_1$ in the base code, etc..

This new set of coding functions is again privacy-preserving and correct. A direct consequence of the definition of the permuted code is reflected on the equivalence in the probability distribution of the random variables 
\begin{align}
&\{\bar{W}_0,\bar{W}_1,\ldots,\bar{W}_{K-1}\},\{\bar{A}^{(0)}_{0},\bar{A}^{(1)}_{0},\ldots,\bar{A}^{(|{\mathcal{Q}}_0|-1)}_0\},\{\bar{A}^{(0)}_{1},\bar{A}^{(1)}_{1},\ldots,\bar{A}^{(|{\mathcal{Q}}_1|-1)}_1\},\ldots,\nonumber\\
&\qquad\qquad\qquad \ldots,\{\bar{A}^{(0)}_{N-1},\bar{A}^{(1)}_{N-1},\ldots,\bar{A}^{(|{\mathcal{Q}}_{N-1}|-1)}_{N-1}\}\nonumber\\
&\stackrel{d}{=}\{W_{\pi^{-1}(0)},W_{\pi^{-1}(1)},\ldots,W_{\pi^{-1}(K-1)}\},\{{A}^{(0)}_{0},{A}^{(1)}_{0},\ldots,{A}^{(|{\mathcal{Q}}_0|-1)}_0\},\{{A}^{(0)}_{1},{A}^{(1)}_{1},\ldots,{A}^{(|{\mathcal{Q}}_1|-1)}_1\},\ldots,\nonumber\\
&\qquad\qquad\qquad\ldots,\{{A}^{(0)}_{N-1},{A}^{(1)}_{N-1},\ldots,{A}^{(|{\mathcal{Q}}_{N-1}|-1)}_{N-1}\},\label{eqn:equivalenceinD}
\end{align}
where $\stackrel{d}{=}$ indicates equivalence in distribution, but not necessarily identical.

Next consider the following code constructed through the space-sharing technique using a base code. Let each message consist of a total of $(K!)L$ symbols, and apply a permuted version of the base code on each length-$L$ sequence (and across $K$ messages), which corresponds to one of the possible permutations on $\{0,1,\ldots,K-1\}$. This space-sharing code is clearly privacy-preserving and correct, however it does not lead to any explicit symmetry relation on the coding rates and the distribution on the queries. It does lead to more subtle invariant relations on the entropies of the subsets of the random variables, {\em e.g., } the joint entropy of a subset of the answers and a subset of the messages is invariant to which subset of messages is being involved. This symmetry cannot produce the invariance on the individual varieties we mentioned earlier.   

\subsection{Variety-symmetry}
\label{sec:variety}

The last symmetry we consider is produced by constructing a different set of queries (and answer varieties) and a new random key $\hat{\mathsf{F}}$ to retrieve the messages. The variety-symmetry is constructed using a different mechanism than the previous two types of symmetry relations. 

Recall in the base code, the random key $\mathsf{F}$ is uniformly distributed on the alphabet $\mathcal{F}$. In the new code, the random key is uniformly distributed on the following set
\begin{align}
\grave{\mathcal{F}}\triangleq \{f_{0:|\mathcal{F}|-1}\in \mathcal{F}^{|\mathcal{F}|}:  f_{0:|\mathcal{F}|-1} \text{ is  a permutation of the elements of } \mathcal{F} \}.
\end{align} 
It follows that $|\grave{\mathcal{F}}|=|\mathcal{F}|!$. The new code operates as follows. The message has $|\mathcal{F}|L$ symbols, which is partitioned into $|\mathcal{F}|$ length-$L$ blocks. Suppose a particular random key realization $\grave{\mathsf{F}}=f_{0:|\mathcal{F}|-1}$ is generated for the new code. For index $i\in \{0,1,\ldots,|\mathcal{F}|-1\}$, the corresponding $i$-th blocks of the messages  are encoded using the base code retrieval strategy determined by the key value $\mathsf{F}=f_i\in \mathcal{F}$.

This new code is clearly correct, and next we show that it is also privacy-preserving. Recall for the request of the message $W_k$, the query for server-$n$ is a deterministic function of the random key $\mathsf{F}=f$ in the base code. Because in the new code, any valid key $f_{0:|\mathcal{F}|-1}$ is a permutation of all the elements in $\mathcal{F}$, the number of times that a particular query $q\in\mathcal{Q}_n$ appears in such a query sequence $f_{0:|\mathcal{F}|-1}$ at server-$n$  is given by
\begin{align}
\kappa_{n,k}(q)\triangleq \left|\{f\in\mathcal{F}: \phi_n(k,f)=q\}\right|. 
\end{align}
Because the base code is privacy-preserving, we have
\begin{align}
\kappa_n(q)\triangleq\kappa_{n,0}(q)=\kappa_{n,1}(q)=\ldots=\kappa_{n,K-1}(q),\quad n\in\{0,1,\ldots,N-1\},\quad q\in \mathcal{Q}_n.\label{eqn:composition}
\end{align}
The composition of any query $q_{0:|\mathcal{F}|-1}$ sent to server-$n$ for the request of the message $W_k$ in this new code, which is a vector of length $|\mathcal{F}|$, is thus given exactly by (\ref{eqn:composition}), and the only difference among the queries is the patterns that these elements in $\mathcal{Q}_n$ are arranged. Thus, the query set at server-$n$ is the constant composition set, {\em i.e., }
\begin{align}
\mathcal{T}_n=\{q_{0:|\mathcal{F}|-1}\in \mathcal{Q}_n^{|\mathcal{F}|}:\text{ the number of appearances of any $q\in\mathcal{Q}_n$ in $q_{0:|\mathcal{F}|-1}$}=\kappa_n(q)\}.
\end{align}
Due to the symmetry in  $\grave{\mathcal{F}}$ and $\mathcal{T}_n$, as well as the uniform distribution on $\grave{\mathcal{F}}$, it is clear that the distribution of the query on $\mathcal{T}_n$ is also uniform, regardless of the identity of the requested message. Thus this new code is indeed privacy-preserving. As a direct consequence of the construction, at each server, all the answer varieties also have the same numbers of symbols to transmit. 

\subsection{Applying the Symmetrization Techniques}


Let us revisit our example for $(N,K)=(2,2)$ given in Section \ref{sec:motivatingexample}. 
To make a variety-symmetric code, we let each message be $2$ bits, denoted as $A=(a_1,a_2),B=(b_1,b_2)$, respectively. The total number of new varieties at each server is $|\grave{\mathcal{Q}}_n|=2!$. This new code is illustrated in Table \ref{table:2_2altersymmetry}. It can be seen that now at each server, the lengths of the answers are indeed the same. We can further apply the server-symmetrization technique, which will produce a code quite similar to that proposed in \cite{sun2017PIRcapacity} and illustrated in Table \ref{table:2_2}.
\begin{table}
\centering
\caption{Answers for message $A$ and $B$ for $(N,K)=(2,2)$ after variety-symmetrization. \label{table:2_2altersymmetry}}
\vspace{0.2cm}
\begin{tabular}{|c||c|c||c|c|}\hline
\multirow{2}{*}{}&\multicolumn{2}{c||}{Requesting $A$}&\multicolumn{2}{c|}{Requesting $B$}\\\cline{2-5}
                            &Server-1&Server-2&Server-1&Server-2\\\hline\hline
$\mathsf{F}=(01)$&$a_2+b_2$  &$a_1,b_2$&$a_2+b_2$  &$a_2,b_1$\\\hline
$\mathsf{F}=(10)$&$a_1+b_1$  &$a_2,b_1$&$a_1+b_1$  &$a_1,b_2$\\\hline
\end{tabular}
\end{table}


We can apply the variety-symmetrization technique on our proposed code with more general parameters.  The message size will increase by a factor of $N^{K-1}$, resulting in a total message size of $N^{K-1}(N-1)$ in the new symmetrized code. In \cite{sun2017optimal}, it was shown that if we insist that the total number of retrieved symbols from all servers is the same for all possible query combinations, then the minimum message size is $N^{K-1}$. Our proposed code in Section \ref{sec:newcode} has a much smaller message size of $N-1$, but does not have this property which turns out to be rather restrictive. On the other hand, the variety-symmetrized code based on our proposed code has a slightly larger message size of  $N^{K-1}(N-1)$ than the optimal value in the restricted setting of \cite{sun2017optimal}. This relatively small increase appears to have stemmed from the decoupled design strategy of applying the symmetrization technique on a base code, instead of designing a symmetric code directly.  

More generally, we can apply all three symmetrization techniques on any asymmetric code (in any order) to obtain a code that is highly symmetric without jeopardizing the retrieval rate, but at the expense of the message size and the upload cost. From this perspective, the reason behind the small message size and upload cost of the proposed code is indeed its asymmetric nature. 

\section{Conclusion}
\label{sec:conclusion}

We proposed a new capacity-achieving PIR code construction, which has the optimal message size and the optimal upload cost. The key to the reduction of both factors, compared to existing constructions, appears to be the asymmetry in the proposed code. In order to prove converse bounds for the optimal message size and the optimal upload cost, we extracted certain critical structures in the converse proof of the PIR capacity. The symmetry structure in the PIR problem is of interest in its own right, and we provided a careful analysis of this structure, which can be used to symmetrize any PIR code into its symmetric version. 

Although in this work we have focused on the most canonical setting of the private information retrieval problem, the proposed code construction using asymmetric structure can be extended to more general settings, such as maximum distance separable code (MDS-coded) databases, which will be reported elsewhere. 

\begin{appendices}

\section{Proof of Technical Lemmas}

\begin{proof}[Proof of Lemma \ref{lemma:first}]
Without loss of generality, let us consider $k=0$. We start by writing the following chain of inequalities:
\begin{align}
&I\left(W_{1:K-1};A^{[0]}_{0:N-1}\Big{|}W_0,\mathsf{F}\right)\nonumber\\
&\stackrel{(a)}{=}I\left(W_{1:K-1};A^{[0]}_{0:N-1},W_0\Big{|}\mathsf{F}\right)\nonumber\\
&=I\left(W_{1:K-1};A^{[0]}_{0:N-1}\Big{|}\mathsf{F}\right)+I\left(W_{1:K-1};W_0\Big{|}A^{[0]}_{0:N-1},\mathsf{F}\right)\nonumber\\
&\stackrel{(b)}{=}I\left(W_{1:K-1};A^{[0]}_{0:N-1}\Big{|}\mathsf{F}\right)\nonumber\\
&=H\left(A^{[0]}_{0:N-1}\Big{|}\mathsf{F}\right)-H\left(A^{[0]}_{0:N-1}\Big{|}W_{1:K-1},\mathsf{F}\right)\nonumber\\
&=H\left(A^{[0]}_{0:N-1}\Big{|}\mathsf{F}\right)-H\left(W_0,A^{[0]}_{0:N-1}\Big{|}W_{1:K-1},\mathsf{F}\right)+H\left(W_0\Big{|}A^{[0]}_{0:N-1},W_{1:K-1},\mathsf{F}\right)\nonumber\\
&\stackrel{(c)}{=}H\left(A^{[0]}_{0:N-1}\Big{|}\mathsf{F}\right)-H\left(W_0|W_{1:K-1},\mathsf{F}\right)\nonumber\\
&\stackrel{(d)}{\leq} \left[\frac{L}{R}-L\right]\log_2|\mathcal{X}|,
\end{align}
where $(a)$ is because the components of $(W_0,W_1,\ldots,W_{K-1},\mathsf{F})$ are mutually independent, $(b)$ and $(c)$ are due to the retrieval correctness requirement and the fact that $A_{{0:N-1}}$ is a deterministic function of $(W_{0:K-1},\mathsf{F})$, and $(d)$ is by the definition of the retrieval rate. 

To see the independence condition \textbf{\textit{P1}}, let us consider $(d)$, and we can write 
\begin{align}
H(A^{[0]}_{0:N-1}|\mathsf{F})&=\sum_{q_{0:N-1}}\mathbf{Pr}(Q^{[0]}_{0:N-1}=q_{0:N-1})H\left(A^{[0]}_{0:N-1}\Big{|}Q^{[0]}_{0:N-1}=q_{0:N-1}\right)\nonumber\\
&=\sum_{q_{0:N-1}}\mathbf{Pr}(Q^{[0]}_{0:N-1}=q_{0:N-1})H\left(A^{(q_0)}_{0},A^{(q_1)}_{1},\ldots,A^{(q_{N-1})}_{N-1}\right)\nonumber\\
&\stackrel{(e)}{\leq} \sum_{q_{0:N-1}}\mathbf{Pr}(Q^{[0]}_{0:N-1}=q_{0:N-1})\sum_{n=0}^{N-1}H\left(A^{(q_n)}_{n}\right)\nonumber\\
&\leq \sum_{q_{0:N-1}}\mathbf{Pr}(Q^{[0]}_{0:N-1}=q_{0:N-1})\sum_{n=0}^{N-1}\ell_n\log_2|\mathcal{Y}|\nonumber\\
&= \log_2|\mathcal{Y}|\sum_{n=0}^{N-1}\mathbb{E}(\ell_n)\nonumber\\
&=\frac{L}{R}\log_2|\mathcal{X}|.
\end{align}
For the equality to hold, it is clear that $(e)$ must be equality for any $q_{0:N-1}$ of non-zero probability, and thus the independence condition \textbf{\textit{P1}} must hold for $k=0$. However, by choosing a permutation $\pi$ on $\{0,1,\ldots,K-1\}$ such that $\pi(k)=0$ and using the same line of proof, it can be concluded that the independence condition holds for all coding matrices of any given requested message. The proof is thus complete.
\end{proof}

\begin{proof}[Proof of Lemma \ref{lemma:second}]
Without loss of generality, let us consider the identity permutation function $\pi(k)=k$. We can start by writing the following chain of information inequalities:
\begin{align}
&NI\left(W_{k:K-1};A_{0:N-1}^{[k-1]}\Big{|}W_{0:k-1},\mathsf{F}\right)\nonumber\\
&\stackrel{(a)}{\geq} \sum_{n=0}^{N-1} I\left(W_{k:K-1};A_{n}^{[k-1]}\Big{|}W_{0:k-1},\mathsf{F}\right)\nonumber\\
&\stackrel{(b)}{=}\sum_{n=0}^{N-1} I\left(W_{k:K-1};A_{n}^{[k]}\Big{|}W_{0:k-1},\mathsf{F}\right)\nonumber\\
&\stackrel{(c)}{=}\sum_{n=0}^{N-1} H\left(A_{n}^{[k]}\Big{|}W_{0:k-1},\mathsf{F}\right)\nonumber\\
&\stackrel{(d)}{\geq}\sum_{n=0}^{N-1} H\left(A_{n}^{[k]}\Big{|}W_{0:k-1},\mathsf{F},A_{0:n-1}^{[k]}\right)\nonumber\\
&=\sum_{n=0}^{N-1} I\left(W_{k:K-1};A_{n}^{[k]}\Big{|}W_{0:k-1},\mathsf{F},A_{0:n-1}^{[k]}\right)\nonumber\\
&=I\left(W_{k:K-1};A_{0:N-1}^{[k]}\Big{|}W_{0:k-1},\mathsf{F}\right)\nonumber\\
&\stackrel{(e)}{=}I\left(W_{k:K-1};W_k,A_{0:N-1}^{[k]}\Big{|}W_{0:k-1},\mathsf{F}\right)\nonumber\\
&=L\log_2|\mathcal{X}|+I\left(W_{k+1:K-1};A_{0:N-1}^{[k]}\Big{|}W_{0:k},\mathsf{F}\right),
\end{align}
where $(c)$ is because the answers are deterministic functions of the messages and the random key $\mathsf{F}$, $(e)$ is due to the retrieval correctness requirement, and the equality $(b)$ can be justified as follows.  We can write that
\begin{align}
&I\left(W_{k:K-1};A_{n}^{[k-1]}\Big{|}W_{0:k-1},\mathsf{F}\right)\nonumber\\
&=H\left(A_{n}^{[k-1]}\Big{|}W_{0:k-1},\mathsf{F}\right)- H\left(A_{n}^{[k-1]}\Big{|}W_{0:K-1},\mathsf{F}\right)\nonumber\\
&=H\left(A_{n}^{[k-1]}\Big{|}W_{0:k-1},\mathsf{F}\right)\nonumber\\
&\stackrel{(f)}{=}H\left(A_{n}^{[k-1]}\Big{|}W_{0:k-1},Q_{n}^{[k-1]}\right)\nonumber\\
&\stackrel{(g)}{=}H\left(A_{n}^{[k]}\Big{|}W_{0:k-1},Q_{n}^{[k]}\right)\nonumber\\
&\stackrel{(h)}{=}H\left(A_{n}^{[k]}\Big{|}W_{0:k-1},\mathsf{F}\right)- H\left(A_{n}^{[k]}\Big{|}W_{0:K-1},\mathsf{F}\right)\nonumber\\
&=I\left(W_{k:K-1};A_{n}^{[k]}\Big{|}W_{0:k-1},\mathsf{F}\right),
\end{align}
where $(f)$ is due to the Markov string $(A_{n}^{[k]},W_{0:K-1})\leftrightarrow Q_{n}^{[k]}\leftrightarrow \mathsf{F}$, $(g)$ is because of the privacy constraint, and $(h)$ is because of the afore-mentioned Markov string and the fact that $Q_{n}^{[k]}$ is a deterministic function of $\mathsf{F}$.

The inequalities $(a)$ and $(d)$ are due to the standard non-negativity property of mutual information. However, the necessary conditions stated in the lemma can be derived from these two inequalities. First consider when $(a)$ is equality, from which we must have for decomposable codes that
\begin{align}
0&=H\left(A_{0:N-1}^{[0]}\Big{|}A_{n}^{[0]},W_{0},\mathsf{F}\right)\nonumber\\
&=\sum_{q_{0:N-1}} \mathbf{Pr}(Q^{[0]}_{0:N-1}=q_{0:N-1})H\left(A_{0:N-1}^{[0]}\Big{|}A_{n}^{[0]},W_{0},Q^{[0]}_{0:N-1}=q_{0:N-1}\right)\nonumber\\
&\stackrel{(i)}{=}\sum_{q_{0:N-1}} \mathbf{Pr}(Q^{[0]}_{0:N-1}=q_{0:N-1})\nonumber\\
&\qquad\cdot H\left(W_{1:K-1}\cdot\left[G^{(q_0)}_{0|0},G^{(q_1)}_{1|0},\ldots,G^{(q_{N-1})}_{N-1|0}\right]\bigg{|}W_{1:K-1}\cdot G^{(q_n)}_{n|0},W_{0},Q^{[0]}_{0:N-1}=q_{0:N-1}\right)\nonumber\\
&\stackrel{(j)}{=}\sum_{q_{0:N-1}} \mathbf{Pr}(Q^{[0]}_{0:N-1}=q_{0:N-1})\cdot H\left(W_{1:K-1}\cdot\left[G^{(q_0)}_{0|0},G^{(q_1)}_{1|0},\ldots,G^{(q_{N-1})}_{N-1|0}\right]\bigg{|}W_{1:K-1}\cdot G^{(q_n)}_{n|0}\right),
\end{align}
where in $(i)$ we have utilized the fact that the component functions $W_0\cdot G^{(q_n)}_{n|1:K-1}$ can be meaningfully subtracted from the answers in the abelian group, and $(j)$ is because $W_0$ is now independent of everything else after the corresponding component functions are eliminated in the answers, and the dependence on $q_{0:N-1}$ is fully absorbed in the answer function matrix $G^{q_{0:N-1}}_n$. 
This implies that for any set of queries $Q_{0:N-1}^{[0]}=q_{0:N-1}$ with a non-zero probability, 
\begin{align}
H\left(W_{1:K-1}\cdot\left[G^{(q_0)}_{0|0},G^{(q_1)}_{1|0},\ldots,G^{(q_{N-1})}_{N-1|0}\right]\bigg{|}W_{1:K-1}\cdot G^{(q_n)}_{n|0}\right)=0,\quad n\in \{0,1,\ldots,N-1\}.
\end{align}
This indeed implies that $W_{1:K-1}\cdot G^{(q_n)}_{n|0}$ can determine any $W_{1:K-1}\cdot G^{(q_{n\rq{}})}_{n\rq{}|0}$, for $n,n\rq{}\in \{0,1,\ldots,N-1\}$. Since the query can be other than for the message $W_0$ (by taking a different permutation $\pi(\cdot)$ in the lemma), it follows that the deterministic property \textit{\textbf{P2}} indeed holds. 

Next consider $(d)$, particularly for $k=K-1$ and the summand for $n=N-1$. For decomposable codes, the inequality being equality implies that 
\begin{align}
0&=I\left(A_{N-1}^{[K-1]};A_{0:N-2}^{[K-1]}\bigg{|}W_{0:K-2},Q_{0:N-1}^{[K-1]}=q_{0:N-1}\right)\nonumber\\
&=I\left(W_{K-1}\cdot G_{N-1|0:K-2}^{(q_{N-1})};W_{K-1}\cdot \left[G_{0|0:K-2}^{(q_{0})},G_{1|0:K-2}^{(q_{1})},\ldots,G_{N-2|0:K-2}^{(q_{N-2})}\right]\right),
\end{align}
which further implies the independence between the random variables $W_{K-1}\cdot G_{N-1|0:K-2}^{(q_{N-1})}$ and $W_{K-1}\cdot \left[G_{0|0:K-2}^{(q_{0})},G_{1|0:K-2}^{(q_{0})},\cdot,G_{N-2|0:K-2}^{(q_{N-2})}\right]$. Since in the above argument, we can choose any value $k$ in $(d)$,  and take any other order in the summation on both sides of $(d)$, indeed the stated independence property \textit{\textbf{P3}} holds. The proof is now complete.
\end{proof}
\end{appendices}

\bibliographystyle{IEEEtran}

\newcommand{\noop}[1]{}

\end{document}